\documentclass[11pt]{article}
\usepackage{makeidx}
\usepackage{euscript}
\usepackage{amsmath,amssymb,amsfonts}
\usepackage{dsfont}
\usepackage{latexsym}
\usepackage{graphicx}
\usepackage{fancybox}
\usepackage{fullpage}
\usepackage[backref]{hyperref}
\usepackage{paralist}
\usepackage{color}
\usepackage{xcolor}
\usepackage{wrapfig}
\usepackage{tikz}
\usepackage{setspace}
\usepackage{algorithm}
\usepackage[noend]{algpseudocode}
\usepackage[framemethod=tikz]{mdframed}
\usepackage{xspace}
\usepackage{pgfplots}
\usepackage{float}
\pgfplotsset{compat=1.14}

\newenvironment{proof}{\noindent{\bf Proof : \ }}{\hfill$\Box$\par\medskip}

\newtheorem{theorem}{Theorem}

\newtheorem{lemma}[theorem]{Lemma}

\newtheorem{definition}[theorem]{Definition}
\newtheorem{problem}{Problem}

\newtheorem{fact}[theorem]{Fact}

\newenvironment{proofof}[1]{\begin{trivlist} \item {\bf Proof
#1:~~}}
  {\qed\end{trivlist}}
\renewenvironment{proofof}[1]{\par\medskip\noindent{\bf Proof of #1: \ }}{\hfill$\Box$\par\medskip}
\newcommand{\namedref}[2]{\hyperref[#2]{#1~\ref*{#2}}}
\newcommand{\thmlab}[1]{\label{thm:#1}}
\newcommand{\thmref}[1]{\namedref{Theorem}{thm:#1}}

\newcommand{\lemlab}[1]{\label{lem:#1}}
\newcommand{\lemref}[1]{\namedref{Lemma}{lem:#1}}

\newcommand{\seclab}[1]{\label{sec:#1}}
\newcommand{\secref}[1]{\namedref{Section}{sec:#1}}

\newcommand{\figlab}[1]{\label{fig:#1}}
\newcommand{\figref}[1]{\namedref{Figure}{fig:#1}}
\newcommand{\alglab}[1]{\label{alg:#1}}
\renewcommand{\algref}[1]{\namedref{Algorithm}{alg:#1}}
\newcommand{\tablelab}[1]{\label{tab:#1}}
\newcommand{\tableref}[1]{\namedref{Table}{tab:#1}}
\newcommand{\deflab}[1]{\label{def:#1}}

\newcommand{\stepref}[1]{\namedref{Step}{step:#1}}
\newcommand{\steplab}[1]{\label{step:#1}}

\newcommand{\probref}[1]{\namedref{Problem}{prob:#1}}
\newcommand{\problab}[1]{\label{prob:#1}}

\newcommand{\COMMENTED}[1]{{}}

\newcommand{\PPr}[1]{\ensuremath{\mathbf{Pr}\left[#1\right]}}
\newcommand{\Ex}[1]{\ensuremath{\mathbf{E}[#1]}}

\newcommand{\eps}{\epsilon}
\newcommand{\estimator}{\ensuremath{\mathsf{Estimator}}}
\newcommand{\oracle}{\ensuremath{\mathsf{Oracle}}}

\newcommand{\bptree}{\ensuremath{\mathsf{BPTree}}}
\newcommand{\countsketch}{\ensuremath{\mathsf{CountSketch}}}
\newcommand{\countmin}{\ensuremath{\mathsf{CountMin}}}
\newcommand{\histogram}{composable histogram}

\newcommand{\smoothcounter}{\ensuremath{\mathsf{SmoothCounter}}}

\newcommand{\HAM}[1]{\ensuremath{\mathsf{HAM}\left(#1\right)}}

\newcommand{\mdef}[1]{{\ensuremath{#1}}\xspace}  
\newcommand{\myfunc}[1]{\mdef{\mathsf{#1}}}      

\DeclareMathOperator*{\polylog}{polylog}

\newcommand{\superscript}[1]{\ensuremath{^{\mbox{\tiny{\textit{#1}}}}}\xspace}
\def \th {\superscript{th}}     
\def \etal{\,{\it et~al.}\,}

\def \polylog  {\mdef{\myfunc{polylog}}}             
\renewcommand{\O}[1]{\ensuremath{\mathcal{O}\left(#1\right)}}						
\newcommand{\T}[1]{\ensuremath{\Theta\left(#1\right)}}								
\newcommand{\D}{\Delta}
\renewcommand{\d}{\delta}

\newcommand{\ig}{\textsf{IndexGreater}}
\newcommand{\gapham}{\textsf{GapHamming}}
\newcommand{\augind}{\textsf{AugmentedIndex}}
\def \lsb        {\mdef{\mathsf{lsb}}}                   

\newcommand{\ignore}[1]{}

\newif\ifnotes\notestrue 
\ifnotes
\newcommand{\elena}[1]{\textcolor{blue}{{\bf (Elena:} {#1}{\bf ) }} \marginpar{\tiny\bf
             \begin{minipage}[t]{0.5in}
               \raggedright E:
                \end{minipage}}}
\newcommand{\harry}[1]{\textcolor{green}{{\bf (Harry:} {#1}{\bf ) }} \marginpar{\tiny\bf
             \begin{minipage}[t]{0.5in}
               \raggedright H:
                \end{minipage}}}
\newcommand{\samson}[1]{\textcolor{purple}{{\bf (Samson:} {#1}{\bf ) }} \marginpar{\tiny\bf
             \begin{minipage}[t]{0.5in}
               \raggedright S:
            \end{minipage}}}
\newcommand{\vova}[1]{\textcolor{red}{{\bf (Vova:} {#1}{\bf ) }} \marginpar{\tiny\bf
             \begin{minipage}[t]{0.5in}
               \raggedright V:
                \end{minipage}}}														
\else
\newcommand{\elena}[1]{}
\newcommand{\harry}[1]{}
\newcommand{\samson}[1]{}
\newcommand{\vova}[1]{}
\fi

\definecolor{mahogany}{rgb}{0.75, 0.25, 0.0}
\definecolor{darkblue}{rgb}{0.0, 0.0, 0.55}
\definecolor{darkpastelgreen}{rgb}{0.01, 0.75, 0.24}
\definecolor{darkgreen}{rgb}{0.0, 0.2, 0.13}
\definecolor{darkgoldenrod}{rgb}{0.72, 0.53, 0.04}
\definecolor{darkred}{rgb}{0.55, 0.0, 0.0}
\definecolor{forestgreen}{rgb}{0.13, 0.55, 0.13}
\hypersetup{
     colorlinks   = true,
     citecolor    = mahogany,
		 linkcolor		= forestgreen
}

\makeatletter
\renewcommand*{\@fnsymbol}[1]{\textcolor{mahogany}{\ensuremath{\ifcase#1\or *\or \dagger\or \ddagger\or
 \mathsection\or \triangledown\or \mathparagraph\or \|\or **\or \dagger\dagger
   \or \ddagger\ddagger \else\@ctrerr\fi}}}
\makeatother

\providecommand{\email}[1]{\href{mailto:#1}{\nolinkurl{#1}\xspace}}

\title{Nearly Optimal Distinct Elements and Heavy Hitters on Sliding Windows}
\author{
Vladimir Braverman\thanks{Department of Computer Science, Johns Hopkins University, Baltimore, MD. 
This material is based upon work supported in part by the National Science Foundation under Grants No. 1447639, 1650041, and 1652257, Cisco faculty award, and by the ONR Award N00014-18-1-2364.
E-mail: \email{vova@cs.jhu.edu}}
\\
Johns Hopkins University
\and
Elena Grigorescu\thanks{Department of Computer Science, Purdue University, West Lafayette, IN. 
Research supported in part by NSF CCF-1649515. 
E-mail: \email{elena-g@purdue.edu}}
\\
Purdue University
\and
Harry Lang\thanks{Department of Mathematics, Johns Hopkins University, Baltimore, MD. 
This material is based upon work supported by the Franco-American Fulbright Commission. 
The author thanks INRIA (l'Institut national de recherche en informatique et en automatique) for hosting him during the writing of this paper.
E-mail: \email{hlang8@jhu.edu}}
\\
Johns Hopkins University
\and
David P. Woodruff\thanks{School of Computer Science, Carnegie Mellon University, Pittsburgh, PA. 
The author would like to acknowledge the support by the National Science Foundation under Grant No. CCF-1815840. 
E-mail: \email{dwoodruf@cs.cmu.edu}}
\\
Carnegie Mellon University
\and
Samson Zhou\thanks{Department of Computer Science, Purdue University, West Lafayette, IN. 
Research supported in part by NSF CCF-1649515. 
E-mail: \email{samsonzhou@gmail.com}}\\
Purdue University
}
\date{}

\begin{document}
\begin{titlepage}
\maketitle
\thispagestyle{empty}
\begin{abstract}
We study the \emph{distinct elements} and \emph{$\ell_p$-heavy hitters} problems in the \emph{sliding window} model, where only the most recent $n$ elements in the data stream form the underlying set. 
We first introduce the \emph{\histogram}, a simple twist on the exponential (Datar \emph{et al.}, SODA 2002) and smooth histograms (Braverman and Ostrovsky, FOCS 2007) that may be of independent interest. 
We then show that the \histogram{} along with a careful combination of existing techniques to track either the identity or frequency of a few specific items suffices to obtain algorithms for both distinct elements and $\ell_p$-heavy hitters that are nearly optimal in both $n$ and $\eps$. 

Applying our new \histogram{} framework, we provide an algorithm that outputs a $(1+\eps)$-approximation to the number of distinct elements in the sliding window model and uses $\O{\frac{1}{\eps^2}\log n\log\frac{1}{\eps}\log\log n+\frac{1}{\eps}\log^2 n}$ bits of space. 
For $\ell_p$-heavy hitters, we provide an algorithm using space $\O{\frac{1}{\eps^p}\log^3 n\left(\log\log n+\log\frac{1}{\eps}\right)}$ for $0<p\le 2$, improving upon the best-known algorithm for $\ell_2$-heavy hitters (Braverman \emph{et al.}, COCOON 2014), which has space complexity $\O{\frac{1}{\eps^4}\log^3 n}$. 
We also show lower bounds of $\Omega\left(\frac{1}{\eps}\log^2 n+\frac{1}{\eps^2}\log n\right)$ for distinct elements and $\Omega\left(\frac{1}{\eps^p}\log^2 n\right)$ for $\ell_p$-heavy hitters, tight up to $\O{\log\frac{1}{\eps}}$ factors for the latter, and tight up to both $\O{\log\log n}$ and $\O{\log\frac{1}{\eps}}$ factors for the former. 
\end{abstract}
\end{titlepage}
\newcommand{\figflow}{
\begin{figure*}[htb]
\centering
\begin{tikzpicture}[scale=2]

\draw[->] (0.2,0)--(0.7,0);
\draw[->] (0,-0.2)--(0,-0.8);
\draw[->] (1,-0.2)--(1,-0.8);

\node at (0,0){$A$};
\node at (1,0){$\D(A)$};
\node at (1,-1){$\hat{f}(\D(A))$};
\node at (0,-1){$f(A)$};
\node at (0.4,-1){$\approx$};
\node at (0,-0.5)[left]{$f$};
\node at (0.5,0)[above]{$\D$};
\node at (1,-0.5)[right]{$\hat{f}$};
\end{tikzpicture}
\caption{The left edge shows direct computation of $f(A)$ from the intput $A$.  Alternatively, the top edge sketches the input to $\D(A)$, which carries enough information to approximate $f(A)$ using the function $\hat{f}$.}\figlab{fig:commutative}
\end{figure*}
}

\newcommand{\figwasteful}{
\begin{figure*}[htb]
\centering
\begin{tikzpicture}[scale=0.4]
\node at (-6.5cm,0.5cm){\tiny{$p_{i-n-2}$}};
\node at (-3.5cm,0.5cm){\tiny{$p_{i-n-1}$}};
\node at (-3.5cm,-0.5cm){\tiny{$p_{i-n-1}$}};
\node at (-0.5cm,0.5cm){\tiny{$p_{i-n}$}};
\node at (-0.5cm,-0.5cm){\tiny{$p_{i-n}$}};
\node at (2.5cm,0.5cm){\tiny{$\ldots$}};
\node at (2.5cm,-0.5cm){\tiny{$\ldots$}};

\draw (-8cm,-0cm) rectangle+(27cm,1cm);
\filldraw[shading=radial,inner color=white, outer color=gray!90, opacity=0.2] (-8cm,0cm) rectangle+(27cm,1cm);
\draw (-5cm,-1cm) rectangle+(24cm,1cm);
\draw (-2cm,-2cm) rectangle+(21cm,1cm);
\filldraw[thick, top color=white,bottom color=red!50!] (-2cm,-2cm) rectangle+(21cm,1cm);
\draw (1cm,-3cm) rectangle+(18cm,1cm);
\filldraw[shading=radial,inner color=white, outer color=gray!90, opacity=0.2] (1cm,-3cm) rectangle+(18cm,1cm);
\draw (4cm,-4cm) rectangle+(15cm,1cm);
\draw (7cm,-5cm) rectangle+(12cm,1cm);
\draw (10cm,-6cm) rectangle+(9cm,1cm);
\filldraw[shading=radial,inner color=white, outer color=gray!90, opacity=0.2] (10cm,-6cm) rectangle+(9cm,1cm);
\draw (13cm,-7cm) rectangle+(6cm,1cm);
\filldraw[shading=radial,inner color=white, outer color=gray!90, opacity=0.2] (13cm,-7cm) rectangle+(6cm,1cm);

\draw[dashed] (-5cm,1cm)--(-5cm,0cm);
\draw[dashed] (-2cm,1cm)--(-2cm,-1cm);
\draw[dashed] (1cm,1cm)--(1cm,-2cm);
\draw[dashed] (16cm,1cm)--(16cm,-7cm);

\draw[dashed] (-2cm,-2cm)--(-2cm,-6.5cm);
\node at(-2cm,-7cm){\tiny{Sliding window begins}};
\node at (-0.5cm,-1.5cm){\tiny{$p_{i-n}$}};
\node at (2.5cm,-1.5cm){\tiny{$\ldots$}};
\foreach \y in {0.5,-0.5,...,-6.5}{
\node at (14.5cm,\y){\tiny{$\ldots$}};
\node at (17.5cm,\y){\tiny{$p_i$}};
}
\end{tikzpicture}
\caption{Each horizontal bar represents an instance of the insertion-only algorithm. The red instance represents the sliding window. Storing an instance beginning at each possible start point would ensure that the exact window is always available, but this requires linear space.
To achieve polylogarithmic space, the histogram stores a strategically chosen set of $\O{\log n}$ instances (shaded grey) so that the value of $f$ on any window can be $(1 + \eps)$-approximated by its value on an adjacent window.}\figlab{pic:Wasteful}
\end{figure*}
}

\newcommand{\figlb}{
\begin{figure*}[htb]
\centering
\begin{tikzpicture}[scale=0.6]
\draw[<->] (-10,4.7) -- (10,4.7);
\draw[decorate,decoration={brace,mirror}](9.8,5) -- (-9.8,5);
\node at (0,5.6){Sliding window string $S$ of length $n$};
\draw (-10,3.6) -- (-10,3.2) -- (-5,3.2) -- (-5,3.6);
\node at (-7.5,4){Block length: $\frac{6\eps n}{\log n}$};
\draw (-5,3.6) -- (-5,3.2) -- (0,3.2) -- (0,3.6);
\node at (-2.5,4){$\frac{6\eps n}{\log n}$};
\draw (-0,3.6) -- (0,3.2) -- (5,3.2) -- (5,3.6);
\node at (2.5,4){$\frac{6\eps n}{\log n}$};
\draw (5,3.6) -- (5,3.2) -- (10,3.2) -- (10,3.6);
\node at (7.5,4){$\frac{6\eps n}{\log n}$};

\foreach \x in {5,5.4,...,6.8}{
	\draw (\x,2.6) -- (\x,2.2) -- (\x+0.4,2.2) -- (\x+0.4,2.6);
}
\foreach \x in {0,0.6,...,3}{
	\draw (\x,2.6) -- (\x,2.2) -- (\x+0.6,2.2) -- (\x+0.6,2.6);
}
\foreach \x in {-5,-4.2,...,-1.4}{
	\draw (\x,2.6) -- (\x,2.2) -- (\x+0.8,2.2) -- (\x+0.8,2.6);
}
\foreach \x in {-6,-7,...,-10}{
	\draw (\x,2.6) -- (\x,2.2) -- (\x+1,2.2) -- (\x+1,2.6);
}

\draw[->] (-3.7,1.6) -- (-3.7,2);
\node at (-2.9,1){Elements $\{0,1,\ldots,(1+2\eps)^i-1\}$ inserted into piece $x_i$ of block $i$.}; 

\node at (-3.9,0){Alice: $x_1\ldots x_m$, where $m=\frac{1}{6\eps}\log n$.};
\draw[decorate,decoration={brace,mirror}](-6.3,-0.6) -- (-5.9,-0.6);
\node at (-4.7,-1.2){Each $x_k$ is $\frac{1}{2}\log n$ bits.};
\end{tikzpicture}
\caption{Construction of distinct elements instance by Alice. Pieces of block $i$ have length $(1+2\eps)^i-1$.}
\figlab{fig:lb}
\end{figure*}
}
\section{Introduction}
The streaming model has emerged as a popular computational model to describe large data sets that arrive sequentially.
In the streaming model, each element of the input arrives one-by-one and algorithms can only access each element once.
This implies that any element that is not explicitly stored by the algorithm is lost forever.
While the streaming model is broadly useful, it does not fully capture the situation in domains where data is time-sensitive such as network monitoring~\cite{C13,CG08,CM05x} and event detection in social media~\cite{OsborneEtAl2014}.
In these domains, elements of the stream appearing more recently are considered more relevant than older elements.
The \textit{sliding window model} was developed to capture this situation~\cite{DatarGIM02}.
In this model, the goal is to maintain computation on only the most recent $n$ elements of the stream, rather than on the stream in its entirety.
We call the most recent $n$ elements \emph{active} and the remaining elements \emph{expired}. 
Any query is performed over the set of active items (referred to as the current window) while ignoring all expired elements.

The problem of identifying the number of distinct elements, is one of the foundational problems in the streaming model.
\begin{problem}[Distinct elements]
\problab{prob:de}
Given an input $S$ of elements in $[m]$, output the number of items $i$ whose frequency $f_i$ satisfies $f_i>0$.
\end{problem}

\noindent
The objective of identifying \emph{heavy hitters}, also known as frequent items, is also one of the most well-studied and fundamental problems. 

\begin{problem}[$\ell_p$-heavy hitters]
\problab{prob:hh}
Given parameters $0<\phi<\eps<1$ and an input $S$ of elements in $[m]$, output all items $i$ whose frequency $f_i$ satisfies $f_i \ge \eps (F_p)^{1/p}$ and no item $i$ for which $f_i\le(\eps-\phi)(F_p)^{1/p}$, where $F_p=\sum_{i \in [m]} f_i^p$. (The parameter $\phi$ is typically assumed to be at least $c\eps$ for some fixed constant $0<c<1$.)
\end{problem}

\noindent
In this paper, we study the distinct elements and heavy hitters problems in the sliding window model. 
We show almost tight results for both problems, using several clean tweaks to existing algorithms. 
In particular, we introduce the \histogram{}, a modification to the exponential histogram~\cite{DatarGIM02} and smooth histogram~\cite{BravermanO07}, that may be of independent interest. 
We detail our results and techniques in the following section, but defer complete proofs to the full version of the paper \cite{BravermanGLWZ18}.

\subsection{Our Contributions} 
\subsubsection*{Distinct elements.}
An algorithm storing $\O{\frac{1}{\eps^2}\log n\log\frac{1}{\delta}(\log\frac{1}{\eps}+\log\log n)}$ bits in the insertion-only model was previously provided~\cite{KaneNW10}. 
Plugging the algorithm into the smooth histogram framework of~\cite{BravermanO07} yields a space complexity of $\O{\frac{1}{\eps^3}\log^3 n(\log\frac{1}{\eps}+\log\log n)}$ bits. 
We improve this significantly as detailed in the following theorem. 
\begin{theorem}
\thmlab{thm:sliding:de}
Given $\eps>0$, there exists an algorithm that, with probability at least $\frac{2}{3}$, provides a $(1+\eps)$-approximation to the number of distinct elements in the sliding window model, using $\O{\frac{1}{\eps^2}\log n\log\frac{1}{\eps}\log\log n+\frac{1}{\eps}\log^2 n}$ bits of space.
\end{theorem}
A known lower bound is $\Omega\left(\frac{1}{\eps^2}+\log n\right)$ bits~\cite{AlonMS99,IndykW03} for insertion-only streams, which is also applicable to sliding windows since the model is strictly more difficult. 
We give a lower bound for distinct elements in the sliding window model, showing that our algorithm is nearly optimal, up to $\log\frac{1}{\eps}$ and $\log\log n$ factors, in both $n$ and $\eps$.
\newcommand{\thmdelb}
{Let $0<\eps\le\frac{1}{\sqrt{n}}$.
Any one-pass sliding window algorithm that returns a $(1+\eps)$-approximation to the number of distinct elements with probability $\frac{2}{3}$ requires $\Omega\left(\frac{1}{\eps}\log^2 n+\frac{1}{\eps^2}\log n\right)$ bits of space.}
\begin{theorem}
\thmlab{thm:de:lb}
\thmdelb
\end{theorem}

\subsubsection*{$\ell_p$-heavy hitters.} 
We first recall in \lemref{lem:tail} a condition that allows the reduction from the problem of finding the $\ell_p$-heavy hitters for $0<p\le 2$ to the problem of finding the $\ell_2$-heavy hitters.  
An algorithm of~\cite{BravermanCINWW17} allows us to maintain an estimate of $F_2$. 
However, observe in \probref{prob:hh} that an estimate for $F_2$ is only part of the problem. 
We must also identify which elements are heavy. 
First, we show how to use tools from~\cite{BravermanCIW16} to find a superset of the heavy hitters. 
This alone is not enough since we may return false-positives (elements such that $f_i < (\eps - \phi) \sqrt{F_2}$). 
By keeping a careful count of the elements (shown in \secref{sec:hh}), we are able to remove these false-positives and obtain the following result, where we have set $\phi = \frac{11}{12} \epsilon$:
\begin{theorem}
\thmlab{thm:sliding:lp}
Given $\eps>0$ and $0<p\le 2$, there exists an algorithm in the sliding window model that, with probability at least $\frac{2}{3}$, outputs all indices $i\in[m]$ for which $f_i\ge\eps F_p^{1/p}$, and reports no indices $i\in[m]$ for which $f_i\le\frac{\eps}{12}F_p^{1/p}$. 
The algorithm has space complexity (in bits) $\O{\frac{1}{\eps^p}\log^3 n\left(\log\log n+\log\frac{1}{\eps}\right)}$.
\end{theorem}
Finally, we obtain a lower bound for $\ell_p$-heavy hitters in the sliding window model, showing that our algorithm is nearly optimal up to $\log\frac{1}{\eps}$ factors in $\eps$.

\newcommand{\thmlb}
{Let $p>0$ and $\eps,\delta \in(0,1)$. 
Any one-pass streaming algorithm that returns the $\ell_p$-heavy hitters in the sliding window model with probability $1-\delta$ requires $\Omega((1-\delta)\eps^{-p}\log^2 n)$ bits of space.}
\begin{theorem}
\thmlab{thm:lb}
\thmlb
\end{theorem}
More details are provided in \secref{sec:hh} and \secref{sec:lb}. 

By standard amplification techniques any result that succeeds with probability $\frac{2}{3}$ can be made to succeed with probability $1-\delta$ while multiplying the space and time complexities by $\O{\log \frac{1}{\delta}}$.
Therefore \thmref{thm:sliding:de} and \thmref{thm:sliding:l2} can be taken with regard to any positive probability of failure.

See \tableref{table1} for a comparison between our results and previous work.

\paragraph{Revision.} 
A previous version of the paper claimed a sliding window algorithm for the $\ell_2$-heavy hitter problem with space complexity $\O{\frac{1}{\eps^2}\log^2 n\left(\log^2\log n+\log\frac{1}{\eps}\right)}$ due to an incorrect union bound over a polylogarithmic number of instances rather than a polynomial number of instances. 
In this version, we correct the statement and analysis, instead using $\O{\frac{1}{\eps^2}\log^3 n\left(\log\log n+\log\frac{1}{\eps}\right)}$ space complexity to preserve the same approach. 

\begin{table}[!htb]
\centering
\resizebox{\columnwidth}{!}{
  \begin{tabular}{| l | c | c |}
    \hline
    Problem & Previous Bound & New Bound \\ \hline
    $\ell_2$-heavy hitters & $\O{\frac{1}{\eps^4}\log^3 n}$~\cite{BravermanGO14} & $\O{\frac{1}{\eps^2}\log^3 n\left(\log\log n+\log\frac{1}{\eps}\right)}$ \\[0.3ex] \hline
    Distinct elements & $\O{\frac{1}{\eps^3}\log^2 n +\frac{1}{\eps}\log^3 n}$~\cite{KaneNW10,BravermanO07} & $\O{\frac{1}{\eps^2}\log\frac{1}{\eps}\log n\log\log n+\frac{1}{\eps}\log^2 n}$ \\[0.3ex]
    \hline
  \end{tabular}
	}
	\caption{Our improvements for $\ell_2$-heavy hitters and distinct elements in the sliding window model.}
	\tablelab{table1}
\end{table}
\subsection{Our Techniques}
We introduce a simple extension of the exponential and smooth histogram frameworks, which use several instances of an underlying streaming algorithm. 
In contrast with the existing frameworks where $\O{\log n}$ different sketches are maintained, we observe in \secref{sec:histogram} when the underlying algorithm has certain guarantees, then we can store these sketches more efficiently. 

\subsubsection*{Sketching Algorithms}
\figwasteful

Consider the sliding window model, where elements eventually expire.
A very simple (but wasteful) algorithm is to simply begin a new instance of the insertion-only algorithm upon the arrival of each new element (\figref{pic:Wasteful}).
The smooth histogram of~\cite{BravermanO07}, summarized in \algref{alg:simple}, shows that storing only $\O{\log n}$ instances suffices.

\begin{algorithm}
\caption{Input: A stream of elements $p_1, p_2, \ldots$ from $[m]$, a window length $n \ge 1$, error $\epsilon \in (0,1)$}\alglab{alg:simple}
\begin{algorithmic}[1]
	\State $T \gets 0$
	\State $i \gets 1$
	\Loop
		\State Get $p_i$ from stream
		\State $T \gets T+1$; $t_T \gets i$; Compute $D(t_T)$, where $\hat{f}(D)$ is a $\left(1\pm\frac{\eps}{4}\right)$-approximation of $f$.
		\For{all $1 < j < T$}
			\If{$\hat{f}(D(t_{j-1}:t_T))<\left(1-\frac{\eps}{4}\right)\hat{f}(D(t_{j+1}:t_T))$}
				\State Delete $t_{j}$; update indices; $T \gets T-1$
			\EndIf
		\EndFor
		\If{$t_2 < i-n$}
			\State Delete $t_{1}$; update indices; $T \gets T-1$
		\EndIf
		\State $i \gets i+1$
	\EndLoop
\end{algorithmic}
\end{algorithm}

\algref{alg:simple} may delete indices for either of two reasons.  
The first (Lines 9-10) is that the index simply expires from the sliding window. 
The second (Lines 7-8) is that the indices immediately before ($t_{j-1}$) and after ($t_{j+1}$) are so close that they can be used to approximate $t_j$.

For the distinct elements problem (\secref{sec:de}), we first claim that a well-known streaming algorithm \cite{Bar-YossefJKST02} provides a $(1+\eps)$-approximation to the number of distinct elements at all points in the stream. 
Although this algorithm is suboptimal for insertion-only streams, we show that it is amenable to the conditions of a \histogram{} (\thmref{thm:main}). 
Namely, we show there is a sketch of this algorithm that is monotonic over suffixes of the stream, and thus there exists an efficient encoding that efficiently stores $D(t_i:t_{i+1})$ for each $1\le i<T$, which allows us to reduce the space overhead for the distinct elements problem. 

For $\ell_2$-heavy hitters (\secref{sec:hh}), we show that the $\ell_2$ norm algorithm of~\cite{BravermanCINWW17} also satisfies the sketching requirement. 
Thus, plugging this into \algref{alg:simple} yields a method to maintain an estimate of $\ell_2$. 
\algref{alg:sliding:l2} uses this subroutine to return the identities of the heavy hitters. 
However, we would still require that all $n$ instances succeed since even $\O{1}$ instances that fail adversarially could render the entire structure invalid by tricking the histogram into deleting the wrong information (see~\cite{BravermanO07} for details).
We show that the $\ell_2$ norm algorithm of~\cite{BravermanCINWW17} actually contains additional structure that only requires the correctness of $\polylog(n)$ instances, thus improving our space usage.

\subsection{Lower Bounds}
\subsubsection*{Distinct elements.}
To show a lower bound of $\Omega\left(\frac{1}{\eps}\log^2 n+\frac{1}{\eps^2}\log n\right)$ for the distinct elements problems, we show in \thmref{thm:de:lb:first} a lower bound of $\Omega\left(\frac{1}{\eps}\log^2 n\right)$ and we show in \thmref{thm:de:lb:second} a lower bound of $\Omega\left(\frac{1}{\eps^2}\log n\right)$. 
We first obtain a lower bound of $\Omega\left(\frac{1}{\eps}\log^2 n\right)$ by a reduction from the \ig{} problem, where Alice is given a string $S=x_1x_2\cdots x_{m}$ and each $x_i$ has $n$ bits so that $S$ has $mn$ bits in total. 
Bob is given integers $i\in[m]$ and $j\in[2^n]$ and must determine whether $x_i>j$ or $x_i\le j$. 

Given an instance of the \ig{} problem, Alice splits the data stream into blocks of size $\O{\frac{\eps n}{\log n}}$ and further splits each block into $\sqrt{n}$ pieces of length $(1+2\eps)^k$, padding the remainder of each block with zeros if necessary. 
For each $i\in[m]$, Alice encodes $x_i$ by inserting the elements $\{0,1,\ldots,(1+2\eps)^k-1\}$ into piece $x_i$ of block $(\ell-i+1)$. 
Thus, the number of distinct elements in each block is much larger than the sum of the number of distinct elements in the subsequent blocks. 
Furthermore, the location of the distinct elements in block $(\ell-i+1)$ encodes $x_i$, so that Bob can recover $x_i$ and compare it with $j$.

We then obtain a lower bound of $\Omega\left(\frac{1}{\eps^2}\log n\right)$ by a reduction from the \gapham{} problem. 
In this problem, Alice and Bob receive length-$n$ bitstrings $x$ and $y$, which have Hamming distance either at least $\frac{n}{2}+\sqrt{n}$ or at most $\frac{n}{2}-\sqrt{n}$, and must decide whether the Hamming distance between $x$ and $y$ is at least $\frac{n}{2}$. 
Recall that for $\eps\le\frac{2}{\sqrt{n}}$, a $(1+\eps)$-approximation can differentiate between at least $\frac{n}{2}+\sqrt{n}$ and at most $\frac{n}{2}-\sqrt{n}$. 
We use this idea to show a lower bound of $\Omega\left(\frac{1}{\eps^2}\log n\right)$ by embedding $\Omega(\log n)$ instances of \gapham{} into the stream. 
As in the previous case, the number of distinct elements corresponding to each instance is much larger than the sum of the number of distinct elements for the remaining instances, so that a $(1+\eps)$-approximation to the number of distinct elements in the sliding window solves the \gapham{} problem for each instance.
\subsubsection*{Heavy hitters.}
To show a lower bound on the problem of finding $\ell_p$-heavy hitters in the sliding window model, we give a reduction from the \augind{} problem. 
Recall that in the \augind{} problem, Alice is given a length-$n$ string $S\in\{1,2\ldots,k\}^n$ (which we write as $[k]^n$) while Bob is given an index $i\in[n]$, as well as $S[1,i-1]$, and must output the $i$\th symbol of the string, $S[i]$. 
To encode $S[i]$ for $S\in[k]^n$, Alice creates a data stream $a_1\circ a_2\circ\ldots\circ a_b$ with the invariant that the heavy hitters in the suffix $a_i\circ a_{i+1}\circ\ldots\circ a_b$ encode $S[i]$. 
Specifically, the heavy hitters in the suffix will be concentrated in the substream $a_i$ and the identities of each heavy hitter in $a_i$ gives a bit of information about the value of $S[i]$. 
To determine $S[i]$, Bob expires the elements $a_1,a_2,\ldots,a_{i-1}$ so all that remains in the sliding window is $a_i\circ a_{i+1}\circ\ldots\circ a_b$, whose heavy hitters encode $S[i]$.

\subsection{Related Work}
The study of the distinct elements problem in the streaming model was initiated by Flajolet and Martin \cite{FlajoletM83} and developed by a long line of work \cite{AlonMS99, GibbonsT01, Bar-YossefJKST02, DurandF03, FlajoletFGM07}. 
Kane, Nelson, and Woodruff \cite{KaneNW10} give an optimal algorithm, using $\O{\frac{1}{\eps^2}+\log n}$ bits of space, for providing a $(1+\eps)$-approximation to the number of distinct elements in a data stream, with constant probability. 
Blasiok \cite{Blasiok18} shows that to boost this probability up to $1-\delta$ for a given $0<\delta<1$, the standard approach of running $\O{\log\frac{1}{\delta}}$ independent instances is actually sub-optimal and gives an optimal algorithm that uses $\O{\frac{\log\delta^{-1}}{\eps^2}+\log n}$ bits of space.

The $\ell_1$-heavy hitters problem was first solved by Misra and Gries, who give a deterministic streaming algorithm using $\O{\frac{1}{\eps}\log n}$ space \cite{MisraG82}. 
Other techniques include the $\countmin$ sketch \cite{CormodeM05}, sticky sampling \cite{MankuM12}, lossy counting \cite{MankuM12}, sample and hold \cite{EstanV03}, multi-stage bloom filters \cite{ChabchoubFM09}, sketch-guided sampling \cite{KumarX06}, and $\countsketch$ \cite{CharikarCF04}. 
Among the numerous applications of the $\ell_p$-heavy hitters problem are network monitoring \cite{DemaineLM02, SenW04}, denial of service prevention \cite{EstanV03, BandiAA07, CormodeKMS08}, moment estimation \cite{IndykW05}, $\ell_p$-sampling \cite{MonemizadehW10}, finding duplicates \cite{GopalanR09}, iceberg queries \cite{FangSGMU98}, and entropy estimation \cite{ChakrabartiCM10, HarveyNO08}.

A stronger notion of ``heavy hitters'' is the \emph{$\ell_2$-heavy hitters}.
This is stronger than the $\ell_1$-guarantee since if $f_i\ge\eps F_1$ then $f^2_i\ge\eps^2 F_1^2\ge\eps^2 F_2$ (and so $f_i\ge\eps\sqrt{F_2}$).  
Thus any algorithm that finds the $\ell_2$-heavy hitters will also find all items satisfying the $\ell_1$-guarantee. 
In contrast, consider a stream that has $f_i=\sqrt{m}$ for some $i$ and $f_j=1$ for all other elements $j$ in the universe. 
Then the $\ell_2$-heavy hitters algorithm will successfully identify $i$ for some constant $\eps$, whereas an algorithm that only provides the $\ell_1$-guarantee requires $\eps=\frac{1}{\sqrt{n}}$, and therefore $\Omega(\sqrt{n}\log n)$ space for identifying $i$.
Moreover, the $\ell_2$-gaurantee is the best we can do in polylogarithmic space, since for $p>2$ it has been shown that identifying $\ell_p$-heavy hitters requires $\Omega(n^{1-2/p})$ bits of space \cite{ChakrabartiKS03, Bar-YossefJKS04}.

The most fundamental data stream setting is the insertion-only model where elements arrive one-by-one.
In the insertion-deletion model, a previously inserted element can be deleted (each stream element is assigned $+1$ or $-1$, generalizing the insertion-only model where only $+1$ is used).
Finally, in the sliding window model, a length $n$ is given and the stream consists only of insertions; points expire after $n$ insertions, meaning that (unlike the insertion-deletion model) the deletions are implicit.
Letting $S = s_1, s_2, \ldots$ be the stream, at time $t$ the frequency vector is built from the window 
$W=\{s_{t-(n-1)},\ldots,s_t\}$ as the active elements, whereas items $\{s_1,\ldots,s_{t-n}\}$ are expired. 
The objective is to identify and report the ``heavy hitters'', namely, the items $i$ for which $f_i$ is large with respect to $W$. 

\tableref{table:models} shows prior work for $\ell_2$-heavy hitters in the various streaming models.
A retuning of $\countsketch$ in \cite{ThorupZ12} solves the problem of $\ell_2$-heavy hitters in $\O{\log^2n}$ bits of space. 
More recently, \cite{BravermanCIW16} presents an $\ell_2$-heavy hitters algorithm using $\O{\log n\log\log n}$ space. 
This algorithm is further improved to an $\O{\log n}$ space algorithm in \cite{BravermanCINWW17}, which is optimal. 

In the insertion-deletion model, $\countsketch$ is space optimal \cite{CharikarCF04, JowhariST11}, but the update time per arriving element is improved by \cite{LarsenNNT16}.
Thus in some sense, the $\ell_2$-heavy hitters problem is completely understood in all regimes except the sliding window model.
We provide a nearly optimal algorithm for this setting, as shown in \tableref{table:models}.

\begin{table}[!htb]
	\centering
	\resizebox{\columnwidth}{!}{
  \begin{tabular}{ | l | c | c | }
    \hline
    Model & Upper Bound & Lower Bound \\ \hline
    Insertion-Only & $\O{\epsilon^{-2}\log n}$~\cite{BravermanCINWW17} & $\Omega(\epsilon^{-2}\log n)$ [Folklore] \\ \hline
    Insertion-Deletion & $\O{\epsilon^{-2}\log^2 n}$~\cite{CharikarCF04} & $\Omega(\epsilon^{-2}\log^2 n)$~\cite{JowhariST11} \\ \hline
    Sliding Windows & $\O{\epsilon^{-2} \log^3 n (\log \epsilon^{-1} + \log \log n)}$ [\thmref{thm:sliding:l2}] & $\Omega(\epsilon^{-2}\log^2 n)$ [\thmref{thm:lb}] \\
    \hline
  \end{tabular}
	}
		\caption{Space complexity in bits of computing $\ell_2$-heavy hitters in various streaming models.  We write $n = |S|$ and to simplify bounds we assume $\log n =\O{\log m}$.}
	\tablelab{table:models} 
\end{table}

We now turn our attention to the sliding window model. 
The pioneering work by Datar \etal\,\cite{DatarGIM02} introduced the exponential histogram as a framework for estimating statistics in the sliding window model. 
Among the applications of the exponential histogram are quantities such as count, sum of positive integers, average, and $\ell_p$ norms. 
Numerous other significant works include improvements to count and sum \cite{GibbonsT02}, frequent itemsets \cite{ChiWYM06}, frequency counts and quantiles \cite{ArasuM04, LeeT06}, rarity and similarity \cite{DatarM02}, variance and $k$-medians \cite{BabcockDMO03} and other geometric problems \cite{FeigenbaumKZ04, ChanS06}. 
Braverman and Ostrovsky \cite{BravermanO07} introduced the smooth histogram as a framework that extends to smooth functions. 
\cite{BravermanO07} also provides sliding window algorithms for frequency moments, geometric mean and longest increasing subsequence. 
The ideas presented by \cite{BravermanO07} also led to a number of other results in the sliding window model \cite{CrouchMS13, BravermanLLM15, BravermanOR15, BravermanLLM16, ChenNZ16, EpastoLVZ17, BravermanDUZ18}. 
In particular, Braverman \etal\,\cite{BravermanGO14} provide an algorithm that finds the $\ell_2$-heavy hitters in the sliding window model with $\phi=c\eps$ for some constant $c>0$, using $\O{\frac{1}{\eps^4}\log^3 n}$ bits of space, improving on results by \cite{HungT08}. 
\cite{Ben-BasatEFK16} also implements and provides empirical analysis of algorithms finding heavy hitters in the sliding window model. 
Significantly, these data structures consider insertion-only data streams for the sliding window model; once an element arrives in the data stream, it remains until it expires. 
It remains a challenge to provide a general framework for data streams that might contain elements ``negative'' in magnitude, or even strict turnstile models. 
For a survey on sliding window algorithms, we refer the reader to \cite{Braverman16}.

\section{Composable Histogram Data Structure Framework}
\seclab{sec:histogram}
We first describe a data structure which improves upon smooth histograms for the estimation of functions with a certain class of algorithms. 
This data structure provides the intuition for the space bounds in \thmref{thm:sliding:de}. 
Before describing the data structure, we need the definition a smooth function.
\begin{definition} \deflab{def:smooth}
\cite{BravermanO07}
A function $f \ge 1$ is $(\alpha,\beta)$-smooth if it has the following properties:
\begin{description}
\item [Monotonicity]
$f(A)\ge f(B)$ for $B\subseteq A$ ($B$ is a suffix of $A$)
\item [Polynomial boundedness]
There exists $c>0$ such that $f(A)\le n^c$.
\item [Smoothness]
For any $\eps\in(0,1)$, there exists $\alpha\in(0,1)$, $\beta\in(0,\alpha]$ so that if $B\subseteq A$ and $(1-\beta)f(A)\le f(B)$, then $(1-\alpha)f(A\cup C)\le f(B\cup C)$ for any adjacent $C$.
\end{description}
\end{definition}
We emphasize a crucial observation made in \cite{BravermanO07}. 
Namely, for $p>1$, $\ell_p$ is a $\left(\eps,\frac{\eps^p}{p}\right)$-smooth function while for $p\le1$, $\ell_p$ is a $(\eps,\eps)$-smooth function.

Given a data stream $S=p_1,p_2,\ldots,p_n$ and a function $f$, let $f(t_1,t_2)$ represent $f$ applied to the substream $p_{t_1},p_{t_1+1},\ldots,p_{t_2}$. 
Furthermore, let $D(t_1:t_2)$ represent the data structure used to approximate $f(t_1,t_2)$. 
\begin{theorem}
\thmlab{thm:main}
Let $f$ be an $(\alpha,\beta)$-smooth function so that $f=\O{n^c}$ for some constant $c$. 
Suppose that for all $\eps,\delta>0$:
\begin{enumerate}
\item
There exists an algorithm $\mathcal{A}$ that maintains at each time $t$ a data structure $D(1:t)$ which allows it to output a value $\hat{f}(1,t)$ so that
 \[\PPr{|\hat{f}(1,t)-f(1,t)|\le\frac{\eps}{2}f(1,t),\,\text{for all }0\le t\le n}\ge 1-\delta.\]
\item
There exists an algorithm $\mathcal{B}$ which, given $D(t_1:t_i)$ and $D(t_i+1:t_{i+1})$, can compute $D(t_i:t_{i+1})$. 
Moreover, suppose storing $D(t_i:t_{i+1})$ uses $\O{g_i(\eps,\delta)}$ bits of space.
\end{enumerate}
Then there exists an algorithm that provides a $(1+\eps)$-approximation to $f$ on the sliding window, using $\displaystyle\O{\frac{1}{\beta}\log^2 n+\sum_{i=1}^{\frac{4}{\beta}\log n}g_i\left(\eps,\frac{\delta}{n}\right)}$ bits of space.
\end{theorem}
We remark that the first condition of \thmref{thm:main} is called ``strong tracking'' and well-motivated by \cite{BlasiokDN17}. 

\section{Distinct Elements}
\seclab{sec:de}
We first show that a well-known streaming algorithm that provides a $(1+\eps)$-approximation to the number of distinct elements actually also provides strong tracking. 
Although this algorithm uses $\O{\frac{1}{\eps^2}\log n}$ bits of space and is suboptimal for insertion-only streams, we show that it is amenable to the conditions of \thmref{thm:main}. 
Thus, we describe a few modifications to this algorithm to provide a $(1+\eps)$-approximation to the number of distinct elements in the sliding window model. 

Define $\lsb(x)$ to be the 0-based index of least significant bit of a non-negative integer $x$ in binary representation. 
For example, $\lsb(10)=1$ and $\lsb(0):=\log(m)$ where we assume $\log(m)=\O{\log n}$. 
Let $S\subset[m]$ and $h:[m]\rightarrow\{0,1\}^{\log m}$ be a random hash function. 
Let $S_k:=\{s\in S:\lsb(h(s))\ge k\}$ so that $2^k|S_k|$ is an unbiased estimator for $|S|$. 
Moreover, for $k$ such that $\Ex{S_k}=\T{\frac{1}{\eps^2}}$, the standard deviation of $2^k|S_k|$ is $\O{\eps|S|}$. 
Let $h_2:[m]\rightarrow[B]$ be a pairwise independent random hash function with $B=\frac{100}{\eps^2}$. 
Let $\Phi_B(m)$ be the expected number of non-empty bins after $m$ balls are thrown at random into $B$ bins so that $\Ex{|h_2(S_k)|}=\Phi_B(|S_k|)$. 
\begin{fact}
$\Phi_m(t)=t\left(1-\left(1-\frac{1}{t}\right)^m\right)$
\end{fact}
\noindent
Blasiok provides an optimal algorithm for a constant factor approximation to the number of distinct elements with strong tracking.
\begin{theorem}\cite{Blasiok18}
There is a streaming algorithm that, with probability $1-\d$, reports a $(1+\eps)$-approximation to the number of distinct elements in the stream after every update and uses $\O{\frac{\log\log n+\log\d^{-1}}{\eps^2}+\log n}$ bits of space.
\end{theorem}
Thus we define an algorithm $\oracle$ that provides a $2$-approximation to the number of distinct elements in the stream after every update, using $\O{\log n}$ bits of space.

Since we can specifically track up to $\O{\frac{1}{\eps^2}}$ distinct elements, let us consider the case where the number of distinct elements is $\omega\left(\frac{1}{\eps^2}\right)$. 
Given access to $\oracle$ to output an estimate $K$, which is a $2$-approximation to the number of distinct elements, we can determine an integer $k>0$ for which $\frac{K}{2^k}=\O{\frac{1}{\eps^2}}$. 
Then the quantity $2^k\Phi^{-1}_B(|h_2(S_k)|)$ provides both strong tracking as well as a $(1+\eps)$-approximation to the number of distinct elements:
\begin{lemma}
\cite{Blasiok18}
\lemlab{lem:phi}
The median of $\O{\log\log n}$ estimators $2^k\Phi^{-1}_B(|h_2(S_k)|)$ is a $(1+\eps)$-approximation at all times for which the number of distinct elements is $\Theta\left(\frac{2^k}{\eps^2}\right)$, with constant probability.
\end{lemma}
Hence, it suffices to maintain $h_2(S_i)$ for each $1\le i\le\log m$, provided access to $\oracle$ to find $k$, and $\O{\log\log n}$ parallel repetitions are sufficient to decrease the variance. 

Indeed, a well-known algorithm for maintaining $h_2(S_i)$ simply keeps a $\log m\times\O{\frac{1}{\eps^2}}$ table $T$ of bits. 
For $0\le i\le\log n$, row $i$ of the table corresponds to $h_2(S_i)$. 
Specifically, the bit in entry $(i,j)$ of $T$ corresponds to $0$ if $h_2(s)\neq j$ for all $s\in S_i$ and corresponds to $1$ if there exists some $s\in S_i$ such that $h_2(s)=j$. 
Therefore, the table maintains $h_2(S_i)$, so then \lemref{lem:phi} implies that the table also gives a $(1+\eps)$-approximation to the number of distinct elements at all times, using $\O{\frac{1}{\eps^2}\log n}$ bits of space and access to $\oracle$. 
Then the total space is $\O{\frac{1}{\eps^2}\log n \log\log n}$ after again using $\O{\log\log n}$ parallel repetitions to decrease the variance.

Na\"ively using this algorithm in the sliding window model would give a space usage dependency of $\O{\frac{1}{\eps^3}\log^2 n \log\log n}$. 
To improve upon this space usage, consider maintaining tables for substreams $(t_1,t), (t_2,t), (t_3,t), \ldots$ where $t_1<t_2<t_3<\ldots<t$. 
Let $T_i$ represent the table corresponding to substream $(t_i,t)$. 
Since $(t_{i+1},t)$ is a suffix of $(t_i,t)$, then the support of the table representing $(t_{i+1},t)$ is a subset of the support of the table representing $(t_i,t)$. 
That is, if the entry $(a,b)$ of $T_{i+1}$ is one, then the entry $(a,b)$ of $T_i$ is one, and similarly for each $j<i$.
Thus, instead of maintaining $\frac{1}{\eps}\log n$ tables of bits corresponding to each of the $(t_i,t)$, it suffices to maintain a single table $T$ where each entry represents the ID of the \emph{last} table containing a bit of one in the entry. 
For example, if the entry $(a,b)$ of $T_9$ is zero but the entry $(a,b)$ of $T_8$ is one, then the entry $(a,b)$ for $T$ is $8$. 
Hence, $T$ is a table of size $\log m\times\O{\frac{1}{\eps^2}}$, with each entry having size $\O{\log\frac{1}{\eps}+\log\log n}$ bits, for a total space of $\O{\frac{1}{\eps^2}\log n\left(\log\frac{1}{\eps}+\log\log n\right)}$ bits. 
Finally, we need $\O{\frac{1}{\eps}\log^2 n}$ bits to maintain the starting index $t_i$ for each of the $\frac{1}{\eps}\log n$ tables represented by $T$. 
Again using a number of repetitions, the space usage is $\O{\frac{1}{\eps^2}\log n\left(\log\frac{1}{\eps}+\log\log n\right)\log\log n +\frac{1}{\eps}\log^2 n}$.

Since this table is simply a clever encoding of the $\O{\frac{1}{\eps}\log n}$ tables used in the smooth histogram data structure, correctness immediately follows. 
We emphasize that the improvement in space follows from the idea of \thmref{thm:main}. 
That is, instead of storing a separate table for each instance of the algorithm in the smooth histogram, we instead simply keep the \emph{difference} between each instance.

Finally, observe that each column in $T$ is monotonically decreasing. 
This is because $S_k:=\{s\in S:\lsb(h(s))\ge k\}$ is a subset of $S_{k-1}$. 
Alternatively, if an item has been sampled to level $k$, it must have also been sampled to level $k-1$. 
Instead of using $\O{\log\frac{1}{\eps}+\log\log n}$ bits per entry, we can efficiently encode the entries for each column in $T$ with the observation that each column is monotonically decreasing. 

\begin{proofof}{\thmref{thm:sliding:de}}
Since the largest index of $T_i$ is $i=\frac{1}{\eps}\log n$ and $T$ has $\log m$ rows, the number of possible columns is $\binom{\frac{1}{\eps}\log n+\log m-1}{\log m}$, which can be encoded using $\O{\log n\log\frac{1}{\eps}}$ bits. 
Correctness follows immediately from \lemref{lem:phi} and the fact that the estimator is monotonic. 
Again we use $\O{\frac{1}{\eps}\log^2 n}$ bits to maintain the starting index $t_i$ for each of the $\frac{1}{\eps}\log n$ tables represented by $T$. 
As $T$ has $\O{\frac{1}{\eps^2}}$ columns and accounting again for the $\O{\log\log n}$ repetitions to decrease the variance, the total space usage is $\O{\frac{1}{\eps^2}\log n\log\frac{1}{\eps}\log\log n+\frac{1}{\eps}\log^2 n}$ bits. 
\end{proofof}

\section{$\ell_p$ Heavy Hitters}
\seclab{sec:hh}
Subsequent analysis by Berinde \etal\,\cite{BerindeICS10} proved that many of the classic $\ell_2$-heavy hitter algorithms not only revealed the identity of the heavy hitters, but also provided estimates of their frequencies. 
Let $f_{tail(k)}$ be the vector $f$ whose largest $k$ entries are instead set to zero. 
Then an algorithm that, for each heavy hitter $i$, outputs a quantity $\hat{f_i}$ such that $|\hat{f_i}-f_i|\le\eps||f_{tail(k)}||_1\le\eps||f||_1$ is said to satisfy the \emph{$(\eps,k)$-tail guarantee}.
Jowhari \etal\,\cite{JowhariST11} show an algorithm that finds the $\ell_2$-heavy hitters and satisfies the tail guarantee can also find the $\ell_p$-heavy hitters. 
Thus, we first show results for $\ell_2$-heavy hitters and then use this property to prove results for $\ell_p$-heavy hitters.

To meet the space guarantees of \thmref{thm:sliding:l2}, we describe an algorithm, \algref{alg:sliding:l2}, that only uses the framework of \algref{alg:simple} to provide a $2$-approximation of the $\ell_2$ norm of the sliding window. 
We detail the other aspects of \algref{alg:sliding:l2} in the remainder of the section.

Recall that \algref{alg:simple} partitions the stream into a series of ``jump-points'' where $f$ increases by a constant multiplicative factor. 
The oldest jump point is before the sliding window and initiates the active window, while the remaining jump points are within the sliding window. 
Therefore, it is possible for some items to be reported as heavy hitters after the first jump point, even though they do not appear in the sliding window at all! 
For example, if the active window has $\ell_2$ norm $2\lambda$, and the sliding window has $\ell_2$ norm $(1+\eps)\lambda$, all $2\eps\lambda$ instances of a heavy hitter in the active window can appear before the sliding window even begins. 
Thus, we must prune the list containing all heavy hitters to avoid the elements with low frequency in the sliding window.

To account for this, we begin a counter for each element immediately after the element is reported as a potential heavy hitter. 
However, the counter must be sensitive to the sliding window, and so we attempt to use a smooth-histogram to count the frequency of each element reported as a potential heavy hitter. 
Even though the count function is $(\eps,\eps)$ smooth, the necessity to track up to $\O{\frac{1}{\eps^2}}$ heavy hitters prevents us from being able to $(1+\eps)$-approximate the count of each element. 
Fortunately, a constant approximation of the frequency of each element suffices to reject the elements whose frequency is less than $\frac{\eps}{8}\ell_2$. 
This additional data structure improves the space dependency to $\O{\frac{1}{\eps^2}}$. 

\subsection{Background for Heavy Hitters}
We now introduce concepts from \cite{BravermanCIW16,BravermanCINWW17} to show the conditions of \thmref{thm:main} apply, first describing an algorithm from \cite{BravermanCINWW17} that provides a good approximation of $F_2$ at all times. 
\begin{theorem}[Remark 8 in \cite{BravermanCINWW17}]
\thmlab{thm:l2}
For any $\eps\in(0,1)$ and $\delta\in[0,1)$, there exists a one-pass streaming algorithm $\estimator$ that outputs at each time $t$ a value $\hat{F}_2^{(t)}$ so that
\[\PPr{|\hat{F}_2^{(t)}-F_2^{(t)}|\le\eps F_2^{(t)},\,\text{for all }0\le t\le n}\ge 1-\delta,\]
and uses $\O{\frac{1}{\eps^2}\log m\left(\log\log m+\log\frac{1}{\eps}\right)\log\frac{1}{\delta}}$ bits of space and $\O{\left(\log\log m+\log\frac{1}{\eps}\right)\log\frac{1}{\delta}}$ update time. 
\end{theorem}
The algorithm of \thmref{thm:l2} is a modified version of the AMS estimator \cite{AlonMS99} as follows. 
Given vectors $Z_j$ of 6-wise independent Rademacher (i.e. uniform $\pm 1$) random variables, let $X_j(t)=\left<Z_j,f^{(t)}\right>$, where $f^{(t)}$ is the frequency vector at time $t$. 
Then \cite{BravermanCINWW17} shows that $Y_t=\frac{1}{N}\sum_{j=1}^N X_{j,t}^2$ is a reasonably good estimator for $F_2$. 
By keeping $X_j(1,t_1), X_j(t_1+1,t_2),\ldots, X_j(t_i+1,t)$, we can compute $X_{j,t}$ from these sketches. 
Hence, the conditions of \thmref{thm:main} are satisfied for $\estimator$, so \algref{alg:simple} can be applied to estimate the $\ell_2$ norm. 
We now refer to a heavy hitter algorithm from \cite{BravermanCINWW17} that is space optimal up to $\log\frac{1}{\eps}$ factors. 
\begin{theorem}[Theorem 11 in \cite{BravermanCINWW17}]
\thmlab{thm:bptree}
For any $\eps>0$ and $\delta\in[0,1)$, there exists a one-pass streaming algorithm, denoted $(\eps,\delta)-\bptree$, that with probability at least $(1-\delta)$, returns a set of $\frac{\eps}{2}$-heavy hitters containing every $\eps$-heavy hitter and an approximate frequency for every item returned satisfying the $(\eps,1/\eps^2)$-tail guarantee. 
The algorithm uses $\O{\frac{1}{\eps^2}\left(\log\frac{1}{\delta\eps}\right)(\log n+\log m)}$ bits of space and has $\O{\log\frac{1}{\delta\eps}}$ update time and $\O{\frac{1}{\eps^2}\log\frac{1}{\delta\eps}}$ retrieval time.
\end{theorem}

Observe that \thmref{thm:l2} combined with \thmref{thm:main} already yields a prohibitively expensive $\frac{1}{\eps^3}$ dependency on $\eps$. 
Thus, we can only afford to set $\eps$ to some constant in \thmref{thm:l2} and have a constant approximation to $F_2$ in the sliding window.

At the conclusion of the stream, the data structure of \thmref{thm:main} has another dilemma: either it reports the heavy hitters for a set of elements $\mathcal{S}_1$ that is a superset of the sliding window or it reports the heavy hitters for a set of elements $\mathcal{S}_2$ that is a subset of the sliding window. 
In the former case, we can report a number of unacceptable false positives, elements that are heavy hitters for $\mathcal{S}_1$ but may not appear at all in the sliding window. 
In the latter case, we may entirely miss a number of heavy hitters, elements that are heavy hitters for the sliding window but arrive before $\mathcal{S}_2$ begins. 
Therefore, we require a separate smooth histogram to track the counter of specific elements.
\begin{theorem}
\thmlab{thm:smooth:counter}
For any $\eps>0$, there exists an algorithm, denoted $(1+\eps)-\smoothcounter$, that outputs a $(1+\eps)$-approximation to the frequency of a given element in the sliding window model, using $\O{\frac{1}{\eps}(\log n+\log m)\log n}$ bits of space.
\end{theorem}
The algorithm follows directly from \thmref{thm:main} and the observation that $\ell_1$ is $(\eps,\eps)$-smooth.
\subsection{$\ell_2$-Heavy Hitters Algorithm}
We now prove \thmref{thm:sliding:l2} using \algref{alg:sliding:l2}. 
We detail our $\ell_2$-heavy hitters algorithm in full, using $\ell_2=\sqrt{F_2}$ and $\eps$-heavy hitters to refer to the $\ell_2$-heavy hitters problem with parameter $\eps$.
\begin{algorithm}[!htb]
\caption{$\eps$-approximation to the $\ell_2$-heavy hitters in a sliding window}
\alglab{alg:sliding:l2}
\textbf{Input:} A stream $S$ of updates $p_i$ for an underlying vector $v$ and a window size $n$.\\
\textbf{Output:} A list including all elements $i$ with $f_i\ge\eps\ell_2$ and no elements $j$ with $f_j<\frac{\eps}{12}\ell_2$.
\begin{algorithmic}[1]
\State{
\steplab{step:one}
Maintain sketches $D(p_{t_1}:p_{t_2}),D(p_{t_2}+1:p_{t_3}),\ldots,D(p_{t_{k-1}}+1:p_{t_k})$ to estimate the $\ell_2$ norm.
}
\Statex{$\vartriangleright$ Use $\estimator$ with parameters $\left(\frac{1}{2},\O{\frac{\delta}{n^4}}\right)$ and \algref{alg:simple} here.
}
\State{
\steplab{step:two}
Let $A_i$ be the merged sketch $D(p_{t_i}+1:p_{t_k})$.
}
\State{
\steplab{step:three}
For each merged sketch $A_i$, find a superset $H_i$ of the $\frac{\eps}{16}$-heavy hitters.
}
\Statex{$\vartriangleright$ Use $\left(\frac{\eps}{16},\frac{\delta}{2}\right)-\bptree$ here. (\thmref{thm:bptree})
}
\State{
\steplab{step:four}
For each element in $H_1$, create a counter.
}
\Statex{$\vartriangleright$ Instantiate a $2-\smoothcounter$ for each of the $\O{\frac{1}{\eps^2}}$ elements reported in $H_1$.
}
\State{
\steplab{step:five}
Let $\hat{\ell}_2$ be the estimated $\ell_2$ norm of $A_1$.
}
\Statex{$\vartriangleright$ Output of $\estimator$ on $A_1$. (\thmref{thm:l2})
}
\State{
\steplab{step:six}
For element $i\in H_1$, let $\hat{f}_i$ be the estimated frequency of $i$.
}
\Statex{$\vartriangleright$ Output by $2-\smoothcounter$.
(\thmref{thm:smooth:counter})
}
\State{
\steplab{step:seven}
Output any element $i$ with $\hat{f}_i\ge\frac{1}{4}\eps\hat{\ell}_2$. 
}
\end{algorithmic}
\end{algorithm}
\begin{lemma}
\lemlab{lem:accept}
Any element $i$ with frequency $f_i>\eps\ell_2$ is output by \algref{alg:sliding:l2}.
\end{lemma}
\begin{proof}
Since the $\ell_2$ norm is a smooth function, and so there exists a smooth-histogram which is an $\left(\frac{1}{2},\frac{\delta}{2}\right)$-estimation of the $\ell_2$ norm of the sliding window by \thmref{thm:main}. 
Thus, $\frac{1}{2}\hat{\ell}_2(A_1)\le \ell_2(W)\le\frac{3}{2}\hat{\ell}_2(A_1)$. 
With probability $1-\frac{\delta}{2}$, any element $i$ whose frequency satisfies $f_i(W)\ge\eps \ell_2(W)$ must have $f_i(W)\ge\eps \ell_2(W)\ge\frac{1}{2}\eps\hat{\ell}_2(A_1)$ and is reported by $\left(\frac{\eps}{16},\frac{\delta}{2}\right)-\bptree$ in \stepref{step:three}.

Since $\bptree$ is instantiated along with $A_1$, the sliding window may begin either before or after $\bptree$ reports each heavy hitter. 
If the sliding window begins after the heavy hitter is reported, then all $f_i(W)$ instances are counted by $\smoothcounter$. 
Thus, the count of $f_i$ estimated by $\smoothcounter$ is at least $f_i(W)\ge\eps \ell_2(W)\ge\frac{1}{2}\eps\hat{\ell}_2(A_1)$, and so \stepref{step:seven} will output $i$.

On the other hand, the sliding window may begin before the heavy hitter is reported.
Recall that the $\bptree$ algorithm identifies and reports an element when it becomes an $\frac{\eps}{16}$-heavy hitter with respect to the estimate of $\ell_2$. 
Hence, there are at most $2\cdot\frac{\eps}{16}\hat{\ell}_2(A_1)\le\frac{1}{8}\eps\hat{\ell}_2(A_1)$ instances of an element appearing in the active window before it is reported by $\bptree$. 
Since $f_i(W)\ge\eps \ell_2(W)\ge\frac{1}{2}\eps\hat{\ell}_2(A_1)$, any element $i$ whose frequency satisfies $f_i(W)\ge\eps \ell_2(W)$ must have
$f_i(W)\ge\frac{\eps}{2}\hat{\ell}_2(A_1)$ and therefore must have at least $\left(\frac{1}{2}-\frac{1}{8}\right)\eps\hat{\ell}_2(A_1)\ge\frac{1}{4}\eps\hat{\ell}_2(A_1)$ instances appearing in the stream after it is reported by $\bptree$. 
Thus, the count of $f_i$ estimated by $\smoothcounter$ is at least $\frac{1}{4}\eps\hat{\ell}_2(A_1)$, and so \stepref{step:seven} will output $i$.
\end{proof}

\begin{lemma}
\lemlab{lem:reject}
No element $i$ with frequency $f_i<\frac{\eps}{12}\ell_2(W)$ is output by \algref{alg:sliding:l2}.
\end{lemma}
\begin{proof}
If $i$ is output by \stepref{step:seven}, then $\hat{f}_i\ge\frac{1}{4}\eps\hat{\ell}_2(A_1)$. 
By the properties of $\smoothcounter$ and $\estimator$, $f_i(W)\ge\frac{\hat{f}_i}{2}\ge\frac{1}{8}\eps\hat{\ell}_2(A_1)\ge\frac{1}{12}\ell_2(W)$, where the last inequality comes from the fact that $\ell_2(W)\le\frac{3}{2}\hat{\ell}_2(A_1)$.
\end{proof}

\begin{theorem}
\thmlab{thm:sliding:l2}
Given $\eps,\delta>0$, there exists an algorithm in the sliding window model (\algref{alg:sliding:l2}) that with probability at least $1-\delta$ outputs all indices $i\in[m]$ for which $f_i\ge\eps\sqrt{F_2}$, and reports no indices $i\in[m]$ for which $f_i\le\frac{\eps}{12}\sqrt{F_2}$. 
The algorithm has space complexity (in bits) $\O{\frac{1}{\eps^2}\log^3 n\left(\log\log n+\log\frac{1}{\eps}\right)}$.
\end{theorem}
\begin{proof}
By \lemref{lem:accept} and \lemref{lem:reject}, \algref{alg:sliding:l2} outputs all elements with frequency at least $\eps\ell_2(W)$ and no elements with frequency less than $\frac{\eps}{12}\ell_2(W)$. 
We now proceed to analyze the space complexity of the algorithm. 
\stepref{step:one} uses \algref{alg:simple} in conjunction with the $\estimator$ routine to maintain a $\frac{1}{2}$-approximation to the $\ell_2$-norm of the sliding window. 
By requiring the probability of failure to be $\O{\frac{\delta}{n^4}}$ in \thmref{thm:l2} and observing that $\beta=\O{1}$ in \thmref{thm:main} suffices for a $\frac{1}{2}$-approximation, it follows that \stepref{step:one} uses $\O{\log n(\log n+\log^2 m\log\log m)}$ bits of space. 
Since \stepref{step:three} runs an instance of $\bptree$ for each of the at most $\O{\log n}$ buckets, then by \thmref{thm:bptree}, it uses $\O{\frac{1}{\eps^2}\left(\log\frac{1}{\delta\eps}\right)\log n(\log n+\log m)}$ bits of space.

Notice that $\bptree$ returns a list of $\O{\frac{1}{\eps^2}}$ elements, by \thmref{thm:bptree}. 
By running $\smoothcounter$ for each of these, \stepref{step:seven} provides a $2$-approximation to the frequency of each element after being returned by $\bptree$. 
By \thmref{thm:smooth:counter}, \stepref{step:seven} has space complexity (in bits) $\O{\frac{1}{\eps^2}(\log n+\log m)\log n}$. 
Assuming $\log m=\O{\log n}$, the algorithm uses $\O{\frac{1}{\eps^2}\log^3 n\left(\log\log n+\log\frac{1}{\eps}\right)}$ bits of space.
\end{proof}

\subsection{Extension to $\ell_p$ norms for $0<p<2$}
To output a superset of the $\ell_p$-heavy hitters rather than the $\ell_2$-heavy hitters, recall that an algorithm provides the $(\eps,k)$-tail guarantee if the frequency estimate $\hat{f_i}$ for each heavy hitter $i\in[m]$ satisfies $|\hat{f_i}-f_i|\le\eps\cdot||f_{tail(k)}||_1$, where $f_{tail(k)}$ is the frequency vector $f$ in which the $k$ most frequent entries have been replaced by zero. 
Jowhari \etal\,\cite{JowhariST11} show the impact of $\ell_2$-heavy hitter algorithms that satisfy the tail guarantee.
\begin{lemma}\cite{JowhariST11}
\lemlab{lem:tail}
For any $p\in(0,2]$, any algorithm that returns the $\eps^{p/2}$-heavy hitters for $\ell_2$ satisfying the tail guarantee also finds the $\eps$-heavy hitters for $\ell_p$.
\end{lemma}
The correctness of \thmref{thm:sliding:lp} immediately follows from \lemref{lem:tail} and \thmref{thm:sliding:l2}.
\begin{proofof}{\thmref{thm:sliding:lp}}
By \thmref{thm:bptree}, $\bptree$ satisfies the tail guarantee. 
Therefore by \lemref{lem:tail}, it suffices to analyze the space complexity of finding the $\eps^{p/2}$-heavy hitters for $\ell_2$. 
By \thmref{thm:sliding:l2}, there exists an algorithm that uses $\O{\frac{1}{\eps^2}\log^3 n\left(\log\log n+\log\frac{1}{\eps}\right)}$ bits of space to find the $\eps$-heavy hitters for $\ell_2$. 
Hence, there exists an algorithm that uses $\O{\frac{1}{\eps^p}\log^3 n\left(\log\log n+\log\frac{1}{\eps}\right)}$ bits of space to find the $\eps$-heavy hitters for $\ell_p$, where $0<p\le 2$.
\end{proofof}

\section{Lower Bounds}
\seclab{sec:lb}
\subsection{Distinct Elements}
To show a lower bound of $\Omega\left(\frac{1}{\eps}\log^2 n+\frac{1}{\eps^2}\log n\right)$ for the distinct elements problem, we show in \thmref{thm:de:lb:first} a lower bound of $\Omega\left(\frac{1}{\eps}\log^2 n\right)$ and we show in \thmref{thm:de:lb:second} a lower bound of $\Omega\left(\frac{1}{\eps^2}\log n\right)$. 
We first obtain a lower bound of $\Omega\left(\frac{1}{\eps}\log^2 n\right)$ by a reduction from the \ig{} problem.
\figlb
\begin{definition}
In the \ig{} problem, Alice is given a string $S=x_1x_2\cdots x_{m}$ of length $mn$, and thus each $x_i$ has $n$ bits.
Bob is given integers $i\in[m]$ and $j\in[2^n]$. 
Alice is allowed to send a message to Bob, who must then determine whether $x_i>j$ or $x_i\le j$.
\end{definition}
Given an instance of the \ig{} problem, Alice first splits the data stream into blocks of size $\O{\frac{\eps n}{\log n}}$. 
She further splits each block into $\sqrt{n}$ pieces of length $(1+2\eps)^k$, before padding the remainder of block $(\ell-k+1)$ with zeros. 
To encode $x_i$ for each $i\in[m]$, Alice inserts the elements $\{0,1,\ldots,(1+2\eps)^k-1\}$ into piece $x_i$ of block $(\ell-i+1)$, before padding the remainder of block $(\ell-k+1)$ with zeros. 
In this manner, the number of distinct elements in each block dominates the number of distinct elements in the subsequent blocks. 
Moreover, the location of the distinct elements in block $(\ell-i+1)$ encodes $x_i$, so that Bob can compare $x_i$ to $j$.

\begin{lemma}
\lemlab{lem:ig}
The one-way communication complexity of \ig{} is $\Omega(nm)$ bits.
\end{lemma}
\begin{proof}
We show the communication complexity of \ig{} through a reduction from the \augind{} problem. 
Suppose Alice is given a string $S\in\{0,1\}^{nm}$ and Bob is given an index $i$ along with the bits $S[1], S[2],\ldots, S[i-1]$. 
Then Bob's task in the \augind{} problem is to determine $S[i]$. 

Observe that Alice can form the string $T=x_1x_2\cdots x_{m}$ of length $mn$, where each $x_k$ has $n$ bits of $S$. 
Alice can then use the \ig{} protocol and communicate to Bob a message that will solve the \ig{} problem. 
Let $j=\left\lfloor\frac{i}{n}\right\rfloor$ so that the symbol $S[i]$ is a bit inside $x_{j+1}$. 
Then Bob constructs the string $w$ by first concatenating the bits $S[jn+1], S[jn+2],\ldots,S[i-1]$, which he is given from the \augind{} problem. 
Bob then appends a zero to $w$, and pads $w$ with ones at the end, until $w$ reaches $n$ bits:
\[w=S[jn+1]\circ S[jn+2]\circ\cdots\circ S[i-1]\circ 0\circ\underbrace{1\circ1\circ\cdots\circ1}_{\text{until $w$ has $n$ bits}}.\]
Bob takes the message from Alice and runs the \ig{} protocol to determine whether $x_j>w$. 
Observe that by construction $x_j>w$ if and only if $S[i]=1$. 
Thus, if the \ig{} protocol succeeds, then Bob will have solved the \augind{} problem, which requires communication complexity $\Omega(nm)$ bits.  
Hence, the communication complexity of \ig{} follows.
\end{proof}

\newcommand{\thmdelbfirst}
{Let $p>0$ and $\eps,\delta \in(0,1)$. 
Any one-pass streaming algorithm that returns a $(1+\eps)$-approximation to the number of distinct elements in the sliding window model with probability $\frac{2}{3}$ requires $\Omega\left(\frac{1}{\eps}\log^2 n\right)$ space.}
\begin{theorem}
\thmlab{thm:de:lb:first}
\thmdelbfirst
\end{theorem}
\begin{proof}
We reduce a one-way communication protocol for \ig{} to finding a $(1+\eps)$-approximation to the number of distinct elements in the sliding window model. 

Let $n$ be the length of the sliding window and suppose Alice receives a string $S=x_1x_2\ldots x_{\ell}\in\{0,1\}^{\ell}$, where $\ell=\frac{1}{6\eps}\log n$ and each $x_k$ has $\frac{1}{2}\log n$ bits. 
Bob receives an index $i\in[\ell]$ and an integer $j\in[\sqrt{n}]$. 
Suppose Alice partitions the sliding window into $\ell$ blocks, each of length $\frac{n}{\ell}=\frac{6\eps n}{\log n}$. 
For each $1\le k\le\frac{1}{6\eps}\log n$, she further splits block $(\ell-k+1)$ into $\sqrt{n}$ pieces of length $(1+2\eps)^k$, before padding the remainder of block $(\ell-k+1)$ with zeros. 
Moreover, for piece $x_k$ of block $(\ell-k+1)$, Alice inserts the elements $\{0,1,\ldots,(1+2\eps)^k-1\}$, before padding the remainder of block $(\ell-k+1)$ with zeros. 
Hence, the sliding window contains all zeros, with the exception of the elements $\{0,1,\ldots,(1+2\eps)^k-1\}$ appearing in piece $x_k$ of block $(\ell-k+1)$ for all $1\le k\le\ell=\frac{1}{6\eps}\log n$. 
Note that $(1+2\eps)^k\le\sqrt[3]{n}$ and $x_k\le\sqrt{n}$ for all $k$, so all the elements fit within each block, which has length $\frac{6\eps n}{\log n}$. 
Finally, Alice runs the $(1+\eps)$-approximation distinct elements sliding window algorithm and passes the state to Bob.
See \figref{fig:lb} for an example of Alice's construction.

Given integers $i\in[\ell]$ and $j\in[\sqrt{n}]$, Bob must determine if $x_i>j$. 
Thus, Bob is interested in $x_i$, so he takes the state of the sliding window algorithm, and inserts a number of zeros to expire each block before block $i$. 
Note that since Alice reversed the stream in her final step, Bob can do this by inserting $(\ell-i)\left(\frac{1}{2}\log n\right)$ number of zeros. 
Bob then inserts $(j-1)(1+2\eps)^i$ additional zeros, to arrive at piece $j$ in block $i$. 
Since piece $x_i$ contains $(1+2\eps)^i$ distinct elements and the remainder of the stream contains $(1+2\eps)^{i-1}$ distinct elements, then the output of the algorithm will decrease below $\frac{(1+2\eps)^i}{1+\eps}$ during piece $x_i$. 
Hence, if the output is less than $\frac{(1+2\eps)^i}{1+\eps}$ after Bob arrives at piece $j$, then $x_i\le j$. 
Otherwise, if the output is at least $\frac{(1+2\eps)^i}{1+\eps}$, then $x_i>j$. 
By the communication complexity of \ig{} (\lemref{lem:ig}), this requires space $\Omega\left(\frac{1}{\eps}\log^2 n\right)$.
\end{proof}
To obtain a lower bound of $\Omega\left(\frac{1}{\eps^2}\log n\right)$, we give a reduction from the \gapham{} problem.
\begin{definition}
\cite{IndykW03}
In the \gapham{} problem, Alice and Bob receive $n$ bit strings $x$ and $y$, which have Hamming distance either at least $\frac{n}{2}+\sqrt{n}$ or at most $\frac{n}{2}-\sqrt{n}$. 
Then Alice and Bob must decide which of these instances is true.
\end{definition}
Chakrabarti and Regev show an optimal lower bound on the communication complexity of \gapham.
\begin{lemma}
\cite{ChakrabartiR12}
The communication complexity of \gapham{} is $\Omega(n)$.
\end{lemma}
Observe that a $(1+\eps)\frac{n}{2}\le\frac{n}{2}+\sqrt{n}$ for $\eps\le\frac{2}{\sqrt{n}}$ and thus a $(1+\eps)$-approximation can differentiate between at least $\frac{n}{2}+\sqrt{n}$ and at most $\frac{n}{2}-\sqrt{n}$. 
We use this idea to show a lower bound of $\Omega\left(\frac{1}{\eps^2}\log n\right)$ by embedding $\Omega(\log n)$ instances of \gapham{} into the stream. 

\newcommand{\thmdelbsecond}
{Let $p>0$ and $\eps,\delta \in(0,1)$. 
Any one-pass streaming algorithm that returns a $(1+\eps)$-approximation to the number of distinct elements in the sliding window model with probability $\frac{2}{3}$ requires $\Omega\left(\frac{1}{\eps^2}\log n\right)$ space for $\eps\le\frac{1}{\sqrt{n}}$.}
\begin{theorem}
\thmlab{thm:de:lb:second}
\thmdelbsecond
\end{theorem}
\begin{proof}
We reduce a one-way communication protocol for the \gapham{} problem to finding a $(1+\eps)$-approximation to the number of distinct elements in the sliding window model. 
For each $\frac{\log\frac{1}{\eps}}{2}\le i\le\frac{\log n-1}{2}$, let $j=2i$ and $x_j$ and $y_j$ each have length $2^j$ and $(x_j,y_j)$ be drawn from a distribution such that with probability $\frac{1}{2}$, $\HAM{x_j,y_j}=(1+4\eps)2^{j-1}$ and otherwise (with probability $\frac{1}{2}$), $\HAM{x_j,y_j}=(1-4\eps)2^{j-1}$. 
Then Alice is given $\{x_j\}$ while Bob is given $\{y_j\}$ and needs to output $\HAM{x_j,y_j}$. 
For $\eps\le\frac{1}{\sqrt{n}}$, this is precisely the hard distribution in the communication complexity of \gapham{} given by \cite{ChakrabartiR12}.

Let $a=\frac{\log\frac{1}{\eps}}{2}$ and $b=\frac{\log n-1}{2}$.
Let $w_{2k}=x_{2k}$ and let $w_{2k-1}$ be a string of length $2^{2k-1}$, all consisting of zeros. 
Suppose Alice forms the concatenated string 
$S=w_{2b}\circ w_{2b-1}\circ\cdots\circ w_{2a+1}\circ w_{2a}$. 
Note that $\sum_{k=2a}^{2b} 2^k\le n$, so $S$ has length less than $n$. 
Alice then forms a data stream by the following process. 
She initializes $k=1$ and continuously increments $k$ until $k=n$. 
At each step, if $S[k]=0$ or $k$ is longer than the length of $S$, Alice inserts a $0$ into the data stream. 
Otherwise, if $S[k]=1$, then Alice inserts $k$ into the data stream. 
Meanwhile, Alice runs the $(1+\eps)$-approximation distinct elements sliding window algorithm and passes the state of the algorithm to Bob.

To find $\HAM{x_{2i},y_{2i}}$, Bob first expires $\left(\sum_{k=2i+1}^{2b} 2^k\right)-2^{2i}$ elements by inserting zeros into the data stream. 
Similar to Alice, Bob initializes $k=1$ and continuously increments $k$ until $k=2^{2i}$. 
At each step, if $y_{2i}[k]=0$ (that is, the $k$\th bit of $y_{2i}$ is zero), then Bob inserts a $0$ into the data stream. 
Otherwise, if $y_{2i}[k]=1$, then Bob inserts $k$ into the data stream. 
At the end of this procedure, the sliding window contains all zeros, nonzero values corresponding to the nonzero indices of the string 
$x_{2i}\circ w_{2i-1}\circ x_{2i-2}\circ\cdots\circ x_{2a+2}\circ w_{2a+1}\circ x_{2a}$,
and nonzero values corresponding to the nonzero indices of $y_{2i}$. 
Observe that each $w_j$ solely consists of zeros and $\sum_{k=a}^{i-1} 2^{2k}<2^{2i-1}$. 
Therefore, $\HAM{x_{2i},y_{2i}}$ is at least $(1-4\eps)2^{2i-1}$ while the number of distinct elements in the sliding window is at most $(1+4\eps)2^{2i}$ while the number of distinct elements in the suffix $x_{2i-2}\circ x_{2i-3}\cdots$ is at most $(1+\eps)2^{2i-2}$.
Thus, a $(1+\eps)$-approximation to the number of distinct elements differentiates between $\HAM{x_{2i},y_{2i}}=(1+4\eps)2^{2i-1}$ and $\HAM{x_{2i},y_{2i}}=(1-4\eps)2^{2i-1}$. 

Since the sliding window algorithm succeeds with probability $\frac{2}{3}$, then the \gapham{} distance problem succeeds with probability $\frac{2}{3}$ across the $\Omega(\log n)$ values of $i$. 
Therefore, any $(1+\eps)$-approximation sliding window algorithm for the number of distinct elements that succeeds with probability $\frac{2}{3}$ requires $\Omega\left(\frac{1}{\eps^2}\log n\right)$ space for $\eps\le\frac{1}{\sqrt{n}}$.
\end{proof}
Hence, \thmref{thm:de:lb} follows from \thmref{thm:de:lb:first} and \thmref{thm:de:lb:second}.

\subsection{$\ell_p$-Heavy Hitters}
To show a lower bound for the $\ell_p$-heavy hitters problem in the sliding window model, we consider the following variant of the \augind{} problem. 
Let $k$ and $n$ be positive integers and $\delta\in[0,1)$.  
Suppose the first player Alice is given a string $S\in[k]^n$, while the second player Bob is given an index $i\in[n]$, as well as $S[1,i-1]$. 
Alice sends a message to Bob, and Bob must output $S[i]$ with probability at least $1-\delta$.
\begin{lemma}\cite{MiltersenNSW95}
\lemlab{lem:cc:lb}
Even if Alice and Bob have access to a source of shared randomness, Alice must send a message of size $\Omega((1-\delta)n\log k)$ in a one-way communication protocol for the \augind{} problem.
\end{lemma}
We reduce the \augind{} problem to finding the $\ell_p$-heavy hitters in the sliding window model. 
To encode $S[i]$ for $S\in[k]^n$, Alice creates a data stream $a_1\circ a_2\circ\ldots\circ a_b$ with the invariant that the heavy hitters in the suffix $a_i\circ a_{i+1}\circ\ldots\circ a_b$ encodes $S[i]$. 
Thus to determine $S[i]$, Bob just needs to run the algorithm for finding heavy hitters on sliding windows and expire the elements $a_1,a_2,\ldots,a_{i-1}$ so all that remains in the sliding window is $a_i\circ a_{i+1}\circ\ldots\circ a_b$.

\begin{proofof}{\thmref{thm:lb}}
We reduce a one-way communication protocol for the \augind{} problem to finding the $\ell_p$ heavy hitters in the sliding window model. 
Let $a=\frac{1}{2^p\eps^p}\log\sqrt{n}$ and $b=\log n$. 
Suppose Alice receives $S=\left[2^a\right]^b$ and Bob receives $i\in[b]$ and $S[1,i-1]$. 
Observe that each $S[i]$ is $\frac{1}{2^p\eps^p}\log\sqrt{n}$ bits and so $S[i]$ can be rewritten as 
$S[i]=w_1\circ w_2\circ\ldots\circ w_t$,
where each $t=\frac{1}{2^p\eps^p}$ and so each $w_i$ is $\log\sqrt{n}$ bits. 

To recover $S[i]$, Alice and Bob run the following algorithm. 
First, Alice constructs data stream $A=a_1\circ a_2\circ\ldots\circ a_b$, which can be viewed as updates to an underlying frequency vector in $\mathbb{R}^n$. 
Each $a_k$ consists of $t$ updates, adding $2^{p(b-k)}$ to coordinates $v_1,v_2\ldots,v_t$ of the frequency vector, where the binary representation of each $v_j\in[n]$ is the concatenation of the binary representation of $j$ with the $\log\sqrt{n}$ bit string $w_j$. 
She then runs the sliding window heavy hitters algorithm and passes the state of the algorithm to Bob. 

Bob expires all elements of the stream before $a_i$, runs the sliding window heavy hitters algorithm on the resulting vector, and then computes the heavy hitters. 
We claim that the algorithm will output $t$ heavy hitters and by concatenating the last $\log\sqrt{n}$ bits of the binary representation of each of these heavy hitters, Bob will recover exactly $S[i]$. 
Observe that the $\ell_p$ norm of the underlying vector represented by $a_i\circ a_{i+1}\circ\ldots\circ a_b$ is exactly
$\left(\frac{1}{2^p\eps^p}(1^p+2^p+4^p+\ldots+2^{p(b-i)})\right)^{1/p}\le\frac{1}{2\eps}2^{b-i+1}=\frac{1}{\eps}2^{b-i}$.
Let $u_1,u_2\ldots,u_t$ be the coordinates of the frequency vector incremented by Alice as part of $a_i$. 
Each coordinate $u_j$ has frequency $2^{b-i}\ge\eps\left(\frac{1}{\eps}2^{b-i}\right)$, so that $u_j$ is an $\ell_p$-heavy hitter. 

Moreover, the first $\log t$ bits of $u_j$ encode $j\in[t]$ while the next $\log\sqrt{n}$ bits encode $w_j$. 
Thus, Bob identifies each heavy hitter and finds the corresponding $j\in[t]$ so that he can concatenate $S[i]=w_1\circ w_2\circ\ldots\circ w_t$.
\end{proofof}

\section*{Acknowledgements}
We would like to thank Nikita Ivkin for pointing out an error in \algref{alg:sliding:l2}.

\def\shortbib{0}
\bibliographystyle{author}
\bibliography{references}
\end{document}